\documentclass[runningheads]{llncs}
\usepackage[T1]{fontenc}
\usepackage{graphicx}

\usepackage{amstext,amsmath,amssymb,color,dsfont}
\usepackage{newfloat,xcolor}
\usepackage{listings}
\usepackage{adjustbox}

\usepackage{caption}
\usepackage{subcaption}

\usepackage{array,enumitem}
\usepackage{makecell,multirow}
\usepackage{comment}
\newtheorem{defn}{Definition}

\newtheorem{thrm}{Theorem}
\newtheorem{prop}{Proposition}

\begin{document}

\title{Towards Inclusive Fairness Evaluation via Eliciting Disagreement Feedback from Non-Expert Stakeholders\thanks{This work is supported by an NSF grant...}}

\author{Mukund Telukunta \and
Venkata Sriram Siddhardh Nadendla}
\authorrunning{Telukunta and Nadendla}
%
\institute{Missouri University of Science and Technology, Rolla, MO, USA\\
\email{\{mt3qb,nadendla\}@umsystem.edu}}
  
\maketitle

\begin{abstract}
Traditional algorithmic fairness notions rely on label feedback, which can only be elicited from expert critics. However, in most practical applications, several non-expert stakeholders also play a major role in the system and can have distinctive opinions about the decision making philosophy. For example, in kidney placement programs, transplant surgeons are very wary about accepting kidney offers for black patients due to genetic reasons. However, non-expert stakeholders in kidney placement programs (e.g. patients, donors and their family members) may misinterpret such decisions from the perspective of social discrimination. This paper evaluates group fairness notions from the viewpoint of non-expert stakeholders, who can only provide binary \emph{agreement/disagreement feedback} regarding the decision in context. Specifically, two types of group fairness notions have been identified: (i) \emph{definite notions} (e.g. calibration), which can be evaluated exactly using disagreement feedback, and (ii) \emph{indefinite notions} (e.g. equal opportunity) which suffer from uncertainty due to lack of label feedback. In the case of indefinite notions, bounds are presented based on disagreement rates, and an estimate is constructed based on established bounds. The efficacy of all our findings are validated empirically on real human feedback dataset.

\end{abstract}

\section{Introduction}
Social discrimination and biases in trained machine learning (ML) models have been investigated extensively using group fairness notions, where the ML-based classifier is compared against ground truth label elicited from expert critics. However, most practical systems comprises of multiple stakeholders with heterogeneous expertise and backgrounds who are either decision-makers, or participants who get impacted by these decisions. Unfortunately, the state-of-the-art algorithmic fairness notions rely on label feedback which can only be elicited from expert stakeholders. This selective feedback elicitation policy has raised concerns amongst other non-expert stakeholders, since their opinions are neglected in the fairness evaluation. Therefore, the main goal of this paper is to develop an inclusive fairness evaluation approach to assess various group fairness notions \cite{hardt2016equality,chouldechova2017fair} in ML-based predictive systems which rely on non-expert stakeholders' feedback. 

However, fairness evaluation through non-expert feedback elicitation comes with many challenges. Firstly, non-expert stakeholders lack expertise due to which they cannot fathom the technical attributes of a given input. Revealing such attributes can only lead to cognitive overloading, thereby discouraging non-expert critics to participate in the fairness evaluation process. Secondly, the ambiguity in non-expert critic feedback increases with the number of classes within the classification problem. For example, in a score-based classifier (e.g. COMPAS' recidivism score predictors), non-expert critics may not decipher the nuances between two scores with little gap between them. Thirdly, most practical systems comprises of several stakeholders with diverse opinions. This makes it economically infeasible to collect feedback using large surveys. 

Due to the aforementioned reasons, this paper propose a simple feedback elicitation model for $M$-ary (magnitude of output space is greater than 2) classifiers, where non-expert stakeholders are requested to reveal binary feedback regarding their agreement/disagreement with the classifier's outcome label. The proposed feedback model does not necessarily align with any one fairness notion. Given such a generalized feedback elicitation model, traditional group fairness notions are assorted broadly into two categories: (i) \emph{definite notions} (e.g. calibration \cite{chouldechova2017fair}), which can be precisely evaluated form disagreement feedback, and (ii) \emph{indefinite notions} (e.g. equal opportunity \cite{hardt2016equality}, predictive parity \cite{chouldechova2017fair}), which can be estimated from disagreement feedback along with given system information. Both upper and lower bounds to indefinite group fairness notions are computed based on elicited disagreement rates. Using these bounds, indefinite group fairness notions are estimated and validated empirically on simulated disagreements constructed using a real dataset collected from 400 crowd workers on the COMPAS system \cite{dressel2018accuracy}. Results demonstrate that the proposed estimates of indefinite group fairness notions based on disagreement feedback, exhibit low error across a wide range of critics in the crowd.



\section{Case Study: Kidney Placement in United States \label{sec: kidney placement}}
Most patients with end-stage renal disease (ESRD) prefer kidney transplantation over long-term dialysis to improve their survival rates. Unfortunately, a significant fraction of procured kidneys are discarded (e.g. 20\% kidneys are discarded in 2019) in the U.S. due to diverse reasons, in spite of severe shortage of kidneys available for transplantation \cite{hart2021optn}. As a result, certain social groups (e.g. African American population, patients with hard-to-place kidneys) are severely disadvantaged in finding a match with deceased donor organs. Various researchers, including the United Network of Organ Sharing (UNOS), have developed recommender systems \cite{mcculloh2023experiment,threlkeld2021reducing,sekercioglu2021machine} using machine learning (ML) algorithms which provides appropriate insights to the experts (e.g. doctors, surgeons) to help make quick and reliable decisions regarding kidney offers and minimize discards. One notable example is the model developed by \cite{bertsimas2017accept}, which predicts the probability of a patient being offered a deceased donor kidney of some quality within some time-frame, given patient characteristics. However, the presence of biases in the data collected at various stages within kidney matching workflow continues to be a major cause for a significantly high kidney discard rate. For example, Kidney Donor Risk Index (KDRI) scores quantify the quality of a deceased donor kidney based on several factors such as age, race, diabetes, and hypertension. However, it also explicitly introduces racial bias against African American donors at the organ procurement organizations (OPOs) during the measurement of donor's kidney quality. Although such a bias is associated with lower allograft or patient survival \cite{callender2008blacks} amongst Black donors, models trained on such data continues to feed the disparate treatment of Black population, thus leading to distrust in such communities. On the other end, the herding behavior is observed at the transplant centers (TXCs), when healthcare professionals blindly reject kidney offers from OPOs if they have been rejected repeatedly in the past. Such unintended biases incorporated into the training workflow results in a discriminatory model regardless of its accuracy. Moreover, many surgeons could have optimism bias, which causes them to believe the availability and promptness of a better-matched organ before they really do \cite{mcculloh2023experiment}. 

In the past, biases in such systems have been investigated based on feedback elicited from experts (e.g. transplant surgeons), who possess enough knowledge regarding the application at hand. On the contrary, it is also essential to consider the opinions of non-expert critics (e.g. patients, donors) who are potential stakeholders and more importantly, the victims of social discrimination. Unlike experts, it is infeasible to elicit feedback from the non-experts due to lack of domain-based knowledge. For example, given the characteristics of a deceased donor kidney (e.g. KDPI, history of diabetes and cancer), it is impossible for a non-expert critic to decide whether or not to accept the kidney for transplantation. In order to address this challenge, we propose a novel feedback elicitation model based on \emph{disagreements}. When presented with the decision made by the system for a specific deceased donor, the non-expert critic can either \emph{agree} or \emph{disagree} with the decision based on their intrinsic and unknown fairness relation. From the provided disagreements we show that various group fairness notions can be estimated.

\section{Preliminaries and Related Work}

\subsection{Group Fairness \label{Sec: group fairness}}
Over the past decade, several statistical group fairness notions have been proposed to measure the biases in a given system. Such fairness notions seek for parity of some statistical measure (e.g. true positive rate, predictive parity value) across all the sensitive attributes (e.g. race, gender) present in the data. Specifically, group fairness notions measure the difference in a specific statistical measure between protected (e.g. Caucasians) and unprotected (e.g. African-Americans) groups of a sensitive attribute. Different versions of group-conditional metrics led to different statistical definitions of fairness \cite{caton2020fairness,chouldechova2018frontiers,mehrabi2021survey,pessach2020algorithmic}. Consider $\mathcal{X}$ and $\mathcal{Y}$ as input and output spaces respectively where, $|\mathcal{Y}| > 2$. Let $y = g(x) \in \mathcal{Y}$ be the outcome label given by the ML-based system for some input $x \in \mathcal{X}$. On the other hand, let $z = f(x)$ be the label given by an alternate classifier (e.g. expert/non-expert critic) for the input $x$. Furthermore, let $\mathcal{X}_m, \mathcal{X}_{m'} \in \mathcal{X}$ denote the protected and unprotected sensitive groups respectively. Inspired from prior work \cite{denis2021fairness}, we define various group fairness notions in case of $M$-ary classification (i.e. $|\mathcal{Y}| > 2$) as follows. 

\vspace{2ex}
\noindent
\textbf{\emph{Statistical parity \cite{dwork2012fairness}}:} This measure seeks to compute the probability difference of individuals who are predicted to be positive across different sensitive groups. Formally, if $SP_{m, k} = \mathbb{P}(y = k \ | \ x \in \mathcal{X}_m)$ denotes the conditional probability of the sensitive group $\mathcal{X}_m$ to receive a label $k \in \mathcal{Y}$, the statistical parity of the system $g$ can be quantified as
\begin{equation}
\max_k \left( \max_{m, m'} \ SP_{m, k} - SP_{m', k} \right),
\label{Eqn: statistical parity}
\end{equation}
for all $k \in \mathcal{Y}$ and $\mathcal{X}_m, \mathcal{X}_{m'} \in \mathcal{X}$. If the difference $SP_{m, k} - SP_{m', k}$ is greater than 0, then the protected group is benefited. On the other hand, if the difference is less than 0, the unprotected group is benefited. Note that the statistical parity of any system can be directly measured from the system's outcome labels, without any need for an alternative classifier. However, this is not the case with other statistical fairness metrics.


\vspace{2ex}
\noindent
\textbf{\emph{Calibration \cite{chouldechova2017fair}}:} A classifier is said to satisfy calibration if both protected and unprotected groups have almost similar positive predictive values (PPV). The PPV represents the probability of an individual with a positive prediction actually experiencing a positive outcome. Formally, if $C_{m,k} = \mathbb{P}(z = k \ | \ y = k, x \in \mathcal{X}_m)$ denote the positive predictive rate for the group $\mathcal{X}_m$, the calibration of the system $g$ with respect to the classifier $f$ is computed as
\begin{equation}
\max_k \left( \max_{m, m'} \ C_{m,k} - C_{m',k} \right),
\label{Eqn: calibration}
\end{equation}
for all $k \in \mathcal{Y}$ and $\mathcal{X}_m, \mathcal{X}_{m'}  \in \mathcal{X}$.


\vspace{2ex}
\noindent
\textbf{\emph{Accuracy Equality \cite{berk2018riskassess}}:} This statistical measure computes the probability that both the classifiers $g$ and $f$ yielding the same label. Specifically, if $AE_{m,k} = \mathbb{P}(y = z \ | \ x \in \mathcal{X}_m)$ denote the conditional probability that both classifiers output the same label for the sensitive group $\mathcal{X}_m$, the accuracy equality of the system $g$ with respect to the alternative classifier $f$ is quantified as
\begin{equation}
\max_k \left( \max_{m, m'} \ AE_{m, k} - AE_{m', k} \right),
\label{Eqn: accuracy equality}
\end{equation}
for all $k \in \mathcal{Y}$ and $\mathcal{X}_m, \mathcal{X}_{m'}  \in \mathcal{X}$.


\vspace{2ex}
\noindent
\textbf{\emph{Equal Opportunity \cite{MoritzOpportunities}}:} A classifier is said to satisfy equal opportunity when both protected and unprotected groups have similar true positive rates (TPR). Formally, if $EO_{m,k} = \mathbb{P}(y = k \ | \ z = k, x \in \mathcal{X}_m)$ denotes the equal opportunity rate for the group $\mathcal{X}_m$, the equal opportunity of the system $g$ with respect to the label $z = f(x)$ is given as
\begin{equation}
\max_k \left( \max_{m, m'} \ EO_{m, k} - EO_{m', k} \right),
\label{Eqn: equalized odds}
\end{equation}
for all $k \in \mathcal{Y}$ and $\mathcal{X}_m, \mathcal{X}_{m'}  \in \mathcal{X}$.


\vspace{2ex}
\noindent
\textbf{\emph{Predictive Equality \cite{corbett2017algorithmic}}:} A classifier is said to satisfy equal opportunity when both protected and unprotected groups have similar false positive rates (FPR). Formally, if $PE_{m,k} = \mathbb{P}(y = k \ | \ z \neq k, x \in \mathcal{X}_m)$ denotes the predictive equality rate of the group $\mathcal{X}_m$, the predictive equality of the system $g$ with respect to the classifier $f$ is computed as
\begin{equation}
\max_k \left( \max_{m, m'} \ PE_{m, k} - PE_{m', k} \right),
\label{Eqn: predictive equality}
\end{equation}
for all $k \in \mathcal{Y}$ and $\mathcal{X}_m, \mathcal{X}_{m'}  \in \mathcal{X}$.


\vspace{2ex}
\noindent
\textbf{\emph{Overall Misclassification Rate \cite{rouzot2022learning}}:}
If $OMR_{m,k} = \mathbb{P}(y \neq k \ | \ z = k, x \in \mathcal{X}_m)$, the overall misclassification rate of the system $g$ with respect the classifier $f$ is given as
\begin{equation}
\max_k \left( \max_{m, m'} \ OMR_{m, k} - OMR_{m', k} \right),
\label{Eqn: overall misclassification}
\end{equation}
for all $k \in \mathcal{Y}$ and $\mathcal{X}_m, \mathcal{X}_{m'}  \in \mathcal{X}$.


\subsection{Survey on Feedback Elicitation for Group Fairness Evaluation}
In the past, several researchers have attempted to model human perception of fairness, but have always tried to fit their revealed feedback to one of the traditional fairness notions. For instance, in an experiment performed by \cite{srivastava2019mathematical}, critics were asked to choose among two different models to identify which notion of fairness (demographic parity or equalized odds) best captures people's perception in the context of both risk assessment and medical applications. Likewise, another team surveyed 502 workers on Amazon's Mturk platform and observed a preference towards \emph{equal opportunity} in \cite{harrison2020empirical}. Dressel and Farid in \cite{dressel2018accuracy} showed that COMPAS is as accurate and fair as that of untrained human auditors regarding predicting recidivism scores. On the other hand, \cite{yaghini2021human} proposed a novel fairness notion, equality of opportunity (EOP), which requires that the distribution of utility should be the same for individuals with similar desert. Based on eliciting human judgments, they learned the proposed EOP notion in terms of criminal risk assessment context. Results show that EOP performs better than existing notions of algorithmic fairness in terms of equalizing utility distribution across groups. Another interesting work is by \cite{grgic2018human}, who discovered that people's fairness concerns are typically multi-dimensional (relevance, reliability, and volitionality), especially when binary feedback was elicited. This means that modeling human feedback should consider several factors beyond social discrimination. A major drawback of these approaches is that the demographics of the participants involved in the experiments \cite{yaghini2021human,grgic2018human,harrison2020empirical,saxena2019fairness} are not evenly distributed. For instance, the conducted experiments ask how models treated Caucasians and African-Americans, but there were insufficient non-Caucasian participants to assess whether there was a relationship between the participant's own demographics and what group was disadvantaged. Moreover, the participants are presented with multiple questions in the existing literature which cannot be scaled for larger decision-based models \cite{yaghini2021human}. Similar efforts have also been carried out in the case of individual fairness notions \cite{dwork2012fairness}. Since individual fairness is beyond the scope of this work, a survey on his topic is omitted for the sake of brevity. Interested readers may refer to \cite{jung2019eliciting,gillen2018online,saxena2019fairness,joseph2016fairness,liu2017calibrated} for more details. 



\section{Non-Expert Disagreement Model}
As defined earlier in Section \ref{Sec: group fairness}, let $\mathcal{X}$ denote the high-dimensional input space where each $x_i \in \mathcal{X}$ represents the input characteristics (e.g. gender, race, history of cancer/history of crime) for all $i = \{1, \cdots, N\}$ input samples. On the other hand, let $\mathcal{Y}$ denote the output space (e.g. donor kidney quality score, decile score) where $|\mathcal{Y}| > 2$. Consider $g: \mathcal{X} \rightarrow \mathcal{Y}$ a ML-based system and $\hat{y} = g(x) \in \mathcal{Y}$ denote the score given to the input profile $x \in \mathcal{X}$. The goal of this paper is to evaluate the social biases present in the system $g$ from the outcome disagreements elicited from a non-expert stakeholder who evaluates using an intrinsic classifier $f$. Without any loss of generality, let $z = f(x)$ denote the intrinsic label of the non-expert critic, who uses an unknown classifier $f$ on the given input $x$. Henceforth, we regard \emph{outcome label} as the label $y$ given by the recommender system, and \emph{intrinsic label} as the unknown label $z$ of the non-expert critic. The non-expert classifier $f$ is typically deemed unreliable since the critics often lack technical knowledge and make amateur judgements. Therefore, this paper assumes that, given an input profile $x$ and outcome label $y = g(x)$, feedback from non-expert critics is elicited in the form of a binary \emph{disagreement} $s \in \{0, 1\}$, as defined below.
\begin{defn}[Non-Expert Disagreement Model]
Given the outcome label of the system $y = g(x)$, the disagreement feedback at the non-expert critic is given by
\begin{equation}
s(y) \ = \ 
\begin{cases}
1, & \text{if } z \neq y,
\\[0.5ex]
0, & \text{otherwise.} 
\end{cases}
\label{Eqn: epsilon-disagreement}
\end{equation}
\label{Defn: Disagreement model}
\end{defn}

If the non-expert critic agrees with the outcome, we assume that the true intrinsic label is similar to the outcome label given by the system i.e. $z = y$. Whereas, if the non-expert critic disagrees with the outcome, the intrinsic label can be any other outcome label. For example, assume that the recidivism tool COMPAS predicts a decile score (usually on a scale of 1-10) of 8 for male, African-American defendant. Assuming that the non-expert disagrees with this outcome given by the COMPAS, his/her true intrinsic outcome may lie anywhere in the range $[1, 8) \cup (8, 10]$. Note that, this uncertainty in the non-expert true intrinsic labels can increase with number of labels in the outcome space.


Furthermore, assume that the input population space $\mathcal{X}$ is partitioned into $M$ groups, namely $X_0, \cdots X_{M-1}$, where $X_0$ represents the non-sensitive group, while all other groups are sensitive in nature. For example, if there are two types of attributes in input profiles, namely gender (male vs. non-male) and race (Caucasian vs. others), $\mathcal{X}$ can be partitioned into $\mathcal{X}_0 \triangleq \mathcal{X}_{M,C}$ (a group of Caucasian males), $\mathcal{X}_1 \triangleq \mathcal{X}_{NM,C}$ (a group of Caucasian non-males), $\mathcal{X}_2 \triangleq \mathcal{X}_{M,O}$ (a group of males from other races) and $\mathcal{X}_3 \triangleq \mathcal{X}_{NM,O}$ (a group of non-males from other races). In such a partition, $\mathcal{X}_0 \triangleq \mathcal{X}_{M,C}$ represents the non-sensitive group, while all other groups are sensitive. Then, the $\epsilon$-disagreement rate with respect to the group $\mathcal{X}_m$ is defined as 
\begin{equation}
\begin{array}{lcl}
DR_m & = & \mathbb{P}( s = 1 \ | \ x \in \mathcal{X}_m)
\\[2ex]
& = & \displaystyle \mathbb{P}(z \neq y \ | \ x \in \mathcal{X}_m)
\end{array}
\end{equation}
where $s$ is the $\epsilon$-disagreement from Equation \eqref{Eqn: epsilon-disagreement}. Similarly, let the conditional probability of \textit{disagreements} for a given outcome label $k \in \mathcal{Y}$ be denoted as
\begin{equation}
\begin{array}{lcl}
DR_{m,k} & = & \mathbb{P}( s = 1 \ | \ y = k, x \in \mathcal{X}_m)
\\[2ex]
& = & \displaystyle \mathbb{P}(z \neq k \ | \ y = k, x \in \mathcal{X}_m)
\end{array}
\end{equation}


\section{Definite Notions}
The set of group fairness notions which can be precisely computed from disagreement rates (and/or statistical parity rates of the system) are identified as definite notions. We determine the notion of accuracy equality \cite{berk2018riskassess} and calibration \cite{chouldechova2017fair} as definite notions, which are quantified as follows.

\subsection{Accuracy Equality} 

\begin{prop}
Given the disagreement rates of the non-expert $DR_{m ,k}, DR_{m' ,k}$ for sensitive groups $\mathcal{X}_m, \mathcal{X}_{m'} \in \mathcal{X}$, the accuracy equality of the system $g$ can be precisely computed as
\begin{equation}
\begin{array}{l}
\max_k \left( \max_{m, m'} \ AE_{m, k} - AE_{m', k} \right)  
\\[2ex]
\qquad \triangleq \max_k \left( \max_{m, m'} \ \displaystyle \sum_{k \in \mathcal{Y}} DR_{m, k}\cdot SP_{m,k} - \displaystyle \sum_{k \in \mathcal{Y}} DR_{m', k}\cdot SP_{m',k} \right),
\end{array}
\end{equation}
for all $k \in \mathcal{Y}$ and $\mathcal{X}_m, \mathcal{X}_{m'} \in \mathcal{X}$.
\end{prop}
\begin{proof}
Using the accuracy equality rate $AE_{m, k}$ for the label $k \in \mathcal{Y}$ and the group $\mathcal{X}_m \in \mathcal{X}$ from the Equation \eqref{Eqn: accuracy equality},
\begin{equation}
AE_{m,k} 
= \mathbb{P}(y = z \ | \ x \in \mathcal{X}_m)
= 1 - \mathbb{P}(y \neq z \ | \ x \in \mathcal{X}_m)
\end{equation}
Considering all the possible labels $k \in \mathcal{Y}$,
\begin{equation}
\begin{array}{lcl}
AE_{m,k} &=& 1 - \displaystyle \sum_{k \in \mathcal{Y}} \mathbb{P}(y = k, z \neq k \ | \ x \in \mathcal{X}_m)
\\[3ex]
&=& 1 - \displaystyle \sum_{k \in \mathcal{Y}} \mathbb{P}(z \neq k \ | \ y = k, x \in \mathcal{X}_m)\cdot \mathbb{P}(y = k \ | \ x \in \mathcal{X}_m)
\end{array}
\end{equation}
Substituting the disagreement rate $DR_{m,k}$ and statistical parity rate $SP_{m,k}$ for the label $k$ and group $\mathcal{X}_m$, we obtain
\begin{equation}
AE_{m,k} = 1 - \displaystyle \sum_{k \in \mathcal{Y}} DR_{m, k}\cdot SP_{m,k}
\end{equation}
\hfill \qed
\end{proof}
 
\subsection{Calibration} 

\begin{prop}
Given the disagreement rates of the non-expert $DR_{m ,k}, DR_{m' ,k}$ for sensitive groups $\mathcal{X}_m, \mathcal{X}_{m'} \in \mathcal{X}$, the calibration of the system $g$ can be precisely computed as
\begin{equation}
\max_k \left( \max_{m, m'} \ C_{m, k} - C_{m', k} \right) \triangleq \max_k \left( \max_{m, m'} \ DR_{m, k} - DR_{m', k} \right)
\end{equation}
for all $k \in \mathcal{Y}$ and $\mathcal{X}_m, \mathcal{X}_{m'} \in \mathcal{X}$.
\end{prop}
\begin{proof}
Using the definition of calibration from the Equation \eqref{Eqn: calibration},
\begin{equation}
\begin{array}{lcl}
|C_{i,k} - C_{j,k}| &=& |1 - DR_{i, k} - 1 + DR_{j, k}|
\\[2ex]
&=& |DR_{i, k} - DR_{j, k}|.
\end{array}
\end{equation}

\hfill \qed
\end{proof}



\section{Indefinite Notions: Bounds and Estimates}
The set of group fairness notions which \emph{cannot} be computed precisely but can be estimated from disagreement rates (along with statistical parity rates) are identified as indefinite notions. We classify the notions of equal opportunity \cite{hardt2016equality}, predictive equality \cite{corbett2017algorithmic}, and overall misclassification rate \cite{rouzot2022learning} as indefinite notions. Given any indefinite group fairness rate $R_{m,k}$ for a given label $k$ under a sensitive group $\mathcal{X}_m$, let any generalized group fairness (GF) notion be defined as
\begin{equation}
GF = \displaystyle \max_k \Big( \max_{m, m'} R_{m,k} - R_{m',k} \Big).
\end{equation}

Then, we have the following estimate for $GF$ based on lower $L_{m,k}$ and upper $U_{m,k}$ bounds computed using disagreement rates.
\begin{thrm}
If an indefinite group fairness rate $R_{m,k}$ is bounded by $L_{m,k}$ from below and $U_{m,k}$ from above, such that both $L_{m,k}$ and $U_{m,k}$ are computed using disagreement feedback. In other words, if
$$L_{m,k} \leq R_{m,k} \leq U_{m,k},$$
then the group fairness notion is bounded by
\begin{equation}
\displaystyle \max_k \Big( \max_{m, m'} L_{m,k} - U_{m',k} \Big) \ \leq \ GF \ \leq \ \displaystyle \max_k \Big( \max_{m, m'} U_{m,k} - L_{m',k} \Big).
\label{Eqn: group fairness bounds}
\end{equation}

\vspace{2ex}
\noindent
Furthermore, an estimate for $GF$ based on these bounds is given by
\begin{equation}
\displaystyle \hat{GF} = \frac{1}{2} \left[ \max_k \Big( \max_{m, m'} L_{m,k} - U_{m',k} \Big) + \max_k \Big( \max_{m, m'} U_{m,k} - L_{m',k} \Big) \right].
\label{Eqn: group fairness estimate}
\end{equation}
\end{thrm}

Unfortunately, the upper bound computation from disagreement rates is found to be non-trivial. Therefore, we consider $U_{m,k} = 1$ for every indefinite notion.

\subsection{Equal Opportunity} 
\begin{prop}
Given the rate of disagreements $DR_{m, k}$ and statistical parity rate $SP_{m,k}$ of the system, the lower bound on equal opportunity of the recommender system $g$ is given by
\begin{equation}
EO_{m, k} \geq \displaystyle \frac{(1 - DR_{m,k}) \cdot SP_{m,k}}{(1 - DR_{m,k}) \cdot SP_{m,k} + \displaystyle \sum_{l \neq k} SP_{m,l}}
\end{equation}
for every sensitive group $\mathcal{X}_m$ and output score $k \in \mathcal{Y}$.
\end{prop}
\begin{proof}
Taking equal opportunity rate of the sensitive group $\mathcal{X}_m$ and expanding it using Bayes theorem.
\begin{equation}
\begin{array}{lcl}
EO_{m, k} & = & \mathbb{P}[y = k \ | \ z = k, x \in \mathcal{X}_m]
\\[3ex]
& = & \displaystyle \frac{\mathbb{P}(y = k, z = k \ | \ x \in \mathcal{X}_m)}{\displaystyle  \mathbb{P}(z = k \ | \ x \in \mathcal{X}_m)}.
\end{array}
\label{Eqn: EO_{m,k}}
\end{equation}
Considering all possible labels $l \in \mathcal{Y}$ in the denominator of Equation \eqref{Eqn: EO_{m,k}}, we obtain the following lower bound:
\begin{equation}
EO_{m, k} = \displaystyle \frac{\mathbb{P}(z = k \ | \ y = k, x \in \mathcal{X}_m) \cdot \mathbb{P}(y = k \ | \ x \in \mathcal{X}_m)}{\displaystyle \sum_{l \in \mathcal{Y}} \mathbb{P}(z = k \ | \ y = l, x \in \mathcal{X}_m) \cdot \mathbb{P}(y = l \ | \ x \in \mathcal{X}_m)}
\end{equation}
Substituting the disagreement rate $DR_{m,k}$ and statistical parity rate $SP_{m,k}$ for the label $k$ and group $\mathcal{X}_m$, we obtain
\begin{equation}
\begin{array}{lcl}
EO_{m, k} & = & \displaystyle \frac{(1 - DR_{m,k}) \cdot SP_{m,k}}{\displaystyle \sum_{l \in \mathcal{Y}} \mathbb{P}(z = k \ | \ y = l, x \in \mathcal{X}_m) \cdot SP_{m,l}}
\\[6ex]
& = & \displaystyle \frac{(1 - DR_{m,k}) \cdot SP_{m,k}}{(1 - DR_{m,k}) \cdot SP_{m,k} + \displaystyle \sum_{l \neq k} \mathbb{P}(z = k \ | \ y = l, x \in \mathcal{X}_m) \cdot SP_{m,l}}
\\[6ex]
& = & \displaystyle \frac{(1 - DR_{m,k}) \cdot SP_{m,k}}{(1 - DR_{m,k}) \cdot SP_{m,k} + \displaystyle \sum_{l \neq k} \mathbb{P}(z = k , y = l \ | \ x \in \mathcal{X}_m)}
\end{array}
\end{equation}
We know that, $\mathbb{P}(z = k , y = l \ | \ x \in \mathcal{X}_m) = \mathbb{P}(z = k \ | \ y = l, x \in \mathcal{X}_m) \cdot SP_{m,l} \leq SP_{m,l}$,
\begin{equation}
EO_{m, k} \geq \displaystyle \frac{(1 - DR_{m,k}) \cdot SP_{m,k}}{(1 - DR_{m,k}) \cdot SP_{m,k} + \displaystyle \sum_{l \neq k} SP_{m,l}}
\end{equation}



\hfill \qed
\end{proof}

Therefore, from Equation \eqref{Eqn: group fairness bounds}, the notion of equal opportunity is bounded by 
\begin{equation}
\displaystyle \max_k \left( \max_{m, m'} \displaystyle \frac{\phi_{m,k}}{\phi_{m,k} + \displaystyle \sum_{l \neq k} SP_{m,l}} - 1 \right) \ \leq \ EO \ \leq \ \displaystyle \max_k \left( \max_{m, m'} 1 -\displaystyle \frac{\phi_{m',k}}{\phi_{m',k} + \displaystyle \sum_{l \neq k} SP_{m',l}} \right)
\end{equation}
and, from Equation \eqref{Eqn: group fairness estimate}, the estimate of equal opportunity from the computed bounds is as follows.
\begin{equation}
\displaystyle \hat{EO} = \frac{1}{2} \left[ \max_k \left( \max_{m, m'}\frac{\phi_{m,k}}{\phi_{m,k} + \displaystyle \sum_{l \neq k} SP_{m,l}} - 1 \right) + \max_k \left( \max_{m, m'} 1 -\displaystyle \frac{\phi_{m',k}}{\phi_{m',k} + \displaystyle \sum_{l \neq k} SP_{m',l}} \right) \right]
\end{equation}
where $\phi_{m,k} = (1 - DR_{m,k}) \cdot SP_{m,k}$.


\subsection{Predictive Equality}
\begin{prop}
Given the rate of disagreements $DR_{m, k}$ and statistical parity rate $SP_{m,k}$ of the system, the lower bound on predictive equality of the recommender system $g$ can be estimated as
\begin{equation}
PE_{m, k} \geq \displaystyle \frac{DR_{m,k}\cdot SP_{m,k}}{DR_{m,k}\cdot SP_{m,k} + 
\displaystyle \sum_{l \neq k} SP_{m,l}}
\end{equation}
for every sensitive group $\mathcal{X}_m$ and output score $k \in \mathcal{Y}$.
\end{prop}
\begin{proof}
Taking the predictive parity rate of the sensitive group $a$
\begin{equation}
\begin{array}{lcl}
PE_{m, k} & = & \mathbb{P}(y = k \ | \ z \neq k, x \in \mathcal{X}_m)
\\[3ex]
&=& \displaystyle \frac{\mathbb{P}(z \neq k \ | \ y =k, x \in \mathcal{X}_m)\cdot \mathbb{P}(y =k \ | \ x \in \mathcal{X}_m)}{\mathbb{P}(z \neq k \ | \ x \in \mathcal{X}_m)}
\end{array}
\end{equation}
Substituting the disagreement rate $DR_{m,k}$ and statistical parity rate $SP_{m,k}$ for the label $k$ and group $\mathcal{X}_m$, we obtain
\begin{equation}
\begin{array}{lcl}
PE_{m,k} &=& \displaystyle \frac{DR_{m,k}\cdot SP_{m,k}}{\mathbb{P}(z \neq k \ | \ x \in \mathcal{X}_m)}
\\[4ex]
&=& \displaystyle \frac{DR_{m,k}\cdot SP_{m,k}}{\displaystyle \sum_{l \in \mathcal{Y}} \mathbb{P}(z \neq k, y = l \ | \ x \in \mathcal{X}_m)}
\\[6ex]
&=& \displaystyle \frac{DR_{m,k}\cdot SP_{m,k}}{DR_{m,k}\cdot SP_{m,k} + 
\displaystyle \sum_{l \neq k} \mathbb{P}(z \neq k, y = l \ | \ x \in \mathcal{X}_m)}
\\[6ex]
&\geq& \displaystyle \frac{DR_{m,k}\cdot SP_{m,k}}{DR_{m,k}\cdot SP_{m,k} + 
\displaystyle \sum_{l \neq k} SP_{m,l}}
\end{array}
\end{equation}
\hfill \qed
\end{proof}

From Equation \eqref{Eqn: group fairness bounds}, the notion of predictive equality is bounded by 
\begin{equation}
\displaystyle \max_k \left( \max_{m, m'} \displaystyle \frac{\mu_{m,k}}{\mu_{m,k} + \sum_{l \neq k} SP_{m,l}} - 1 \right) \ \leq \ PE \ \leq \ \displaystyle \max_k \left( \max_{m, m'} 1 - \displaystyle \frac{\mu_{m',k}}{\mu_{m',k} + \sum_{l \neq k} SP_{m',l}} \right)
\end{equation}
and, from Equation \eqref{Eqn: group fairness estimate}, the estimate of predictive equality from the computed bounds is as follows.
\begin{equation}
\displaystyle \hat{PE} = \frac{1}{2} \left[ \displaystyle \max_k \left( \max_{m, m'} \displaystyle \frac{\mu_{m,k}}{\mu_{m,k} + \sum_{l \neq k} SP_{m,l}} - 1 \right) + \displaystyle \max_k \left( \max_{m, m'} 1 - \displaystyle \frac{\mu_{m',k}}{\mu_{m',k} + \sum_{l \neq k} SP_{m',l}} \right) \right]
\end{equation}
where $\mu_{m,k} = DR_{m,k} \cdot SP_{m,k}$.

\subsection{Overall Misclassification Rate}

\begin{prop}
Given the rate of disagreements $DR_{m, k}$ and statistical parity rate $SP_{m,k}$ of the system, the lower bound on overall misclassification rate of the recommender system $g$ is given by
\begin{equation}
OMR_{m, k} \geq \displaystyle \frac{\displaystyle \sum_{l \neq k} SP_{m,l}}{(1 - DR_{m,k}) \cdot SP_{m,k} + \displaystyle \sum_{l \neq k} SP_{m,l}}
\end{equation}
for every sensitive group $\mathcal{X}_m$ and output score $k \in \mathcal{Y}$.
\end{prop}
\begin{proof}
From the definition of equal opportunity, we have
\begin{equation}
EO_{m,k} \geq \displaystyle \frac{(1 - DR_{m,k}) \cdot SP_{m,k}}{(1 - DR_{m,k}) \cdot SP_{m,k} + \displaystyle \sum_{l \neq k} SP_{m,l}}
\end{equation}
We know that, $OMR_{m,k} = 1 - EO_{m,k}$. Therefore we have,
\begin{equation}
\begin{array}{lcl}
OMR_{m,k} & \geq & 1 - \displaystyle \frac{(1 - DR_{m,k}) \cdot SP_{m,k}}{(1 - DR_{m,k}) \cdot SP_{m,k} + \displaystyle \sum_{l \neq k} SP_{m,l}}
\\[6ex]
& = & \displaystyle \frac{\displaystyle \sum_{l \neq k} SP_{m,l}}{(1 - DR_{m,k}) \cdot SP_{m,k} + \displaystyle \sum_{l \neq k} SP_{m,l}}
\end{array}
\end{equation}
\hfill \qed
\end{proof}

Therefore, from Equation \eqref{Eqn: group fairness bounds}, the notion of overall misclassification is bounded by 
\begin{equation}
\displaystyle \max_k \left( \max_{m, m'} \displaystyle \frac{\Omega_{m,k}}{\phi_{m,k} + \Omega_{m,k}} - 1 \right) \ \leq \ OMR \ \leq \ \displaystyle \max_k \left( \max_{m, m'} 1 - \displaystyle \frac{\Omega_{m',k}}{\phi_{m',k} + \Omega_{m',k}} \right)
\end{equation}
and, from Equation \eqref{Eqn: group fairness estimate}, the estimate of overall misclassification from the computed bounds is as follows.
\begin{equation}
\displaystyle \hat{PE} = \frac{1}{2} \left[ \displaystyle \max_k \left( \max_{m, m'} \displaystyle \frac{\Omega_{m,k}}{\phi_{m,k} + \Omega_{m,k}} - 1 \right) + \displaystyle \max_k \left( \max_{m, m'} 1 - \displaystyle \frac{\Omega_{m',k}}{\phi_{m',k} + \Omega_{m',k}} \right) \right]
\end{equation}
where $\Omega_{m,k} =\displaystyle \sum_{l \neq k} SP_{m,l}$ and $\phi_{m,k} = (1 - DR_{m,k}) \cdot SP_{m,k}$.

\begin{figure}[!t]
\centering
\begin{subfigure}[b]{0.49\textwidth}
\centering
\includegraphics[width=\textwidth]{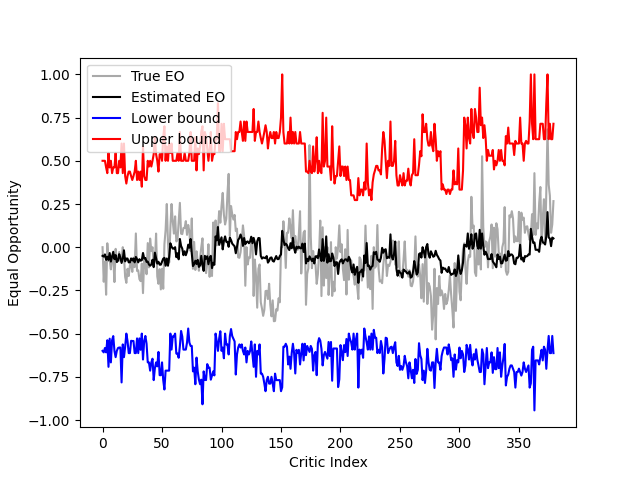}
\caption{Equal Opportunity}
\label{Fig: Estimated EO}
\end{subfigure}
\hfill
\begin{subfigure}[b]{0.49\textwidth}
\centering
\includegraphics[width=\textwidth]{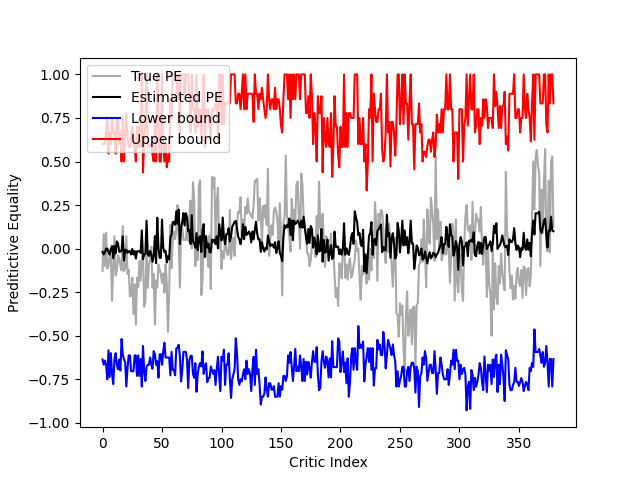}
\caption{Predictive Equality}
\label{Fig: Estimated PE}
\end{subfigure}
\hfill
\begin{subfigure}[b]{0.49\textwidth}
\centering
\includegraphics[width=\textwidth]{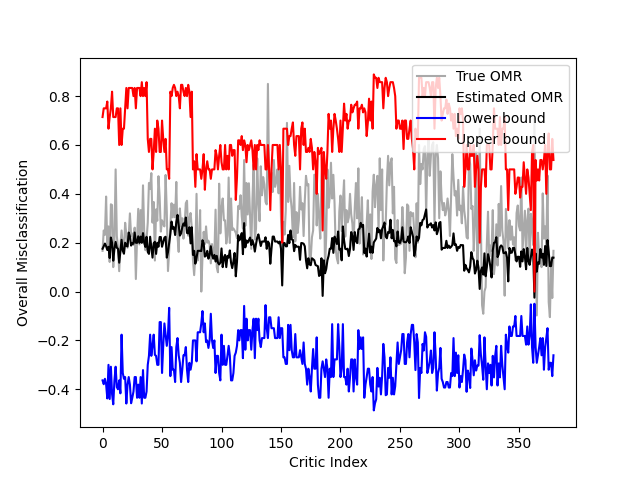}
\caption{Overall Misclassification Rate}
\label{Fig: Estimated OM}
\end{subfigure}
\hfill
\begin{subfigure}[b]{0.49\textwidth}
\centering
\includegraphics[width=\textwidth]{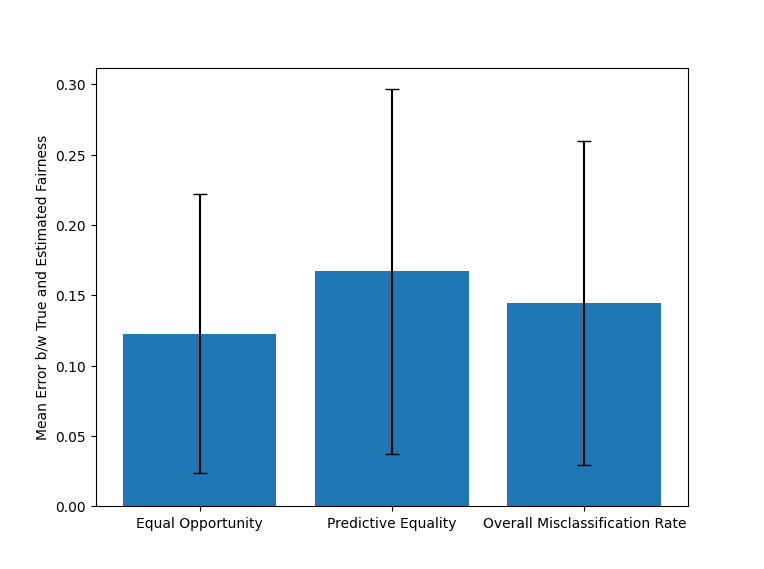}
\caption{Mean Error of all Estimates}
\label{Fig: Error in estimation}
\end{subfigure}
\caption{Comparison of True Group Fairness (Grey), Proposed Estimate (Black), Upper (Red) and Lower bounds (Blue), and their respective mean error averaged over all 400 critics}
\end{figure}


\section{Validation Methodology and Results}

We validate our theoretical findings using the real human feedback collected by Dressel and Farid \cite{dressel2018accuracy}. This data acquisition experiment consists of a short description of the defendant (gender, age, race, and previous criminal history) is provided to the human critics. A total of 1000 defendant descriptions are used that are drawn randomly from the original ProPublica's COMPAS dataset. Furthermore, these descriptions were divided into 20 subsets of 50 each. The experiment consisted of 400 different critics recruited from Amazon Mechanical Turk and each one of them was randomly assigned to see one of these 20 subsets. The participants were then asked to respond either \emph{yes} or \emph{no} to the question “Do you think this person will commit another crime within 2 years?”. From these responses, we process the dataset to obtain disagreement feedback of each critic. Since the responses are binary, the disagreement feedback can be extracted as follows $s = y \oplus z$ where, $y$ denote COMPAS outcome label, $z$ denote the critic's response, and $\oplus$ represents XOR operation. In other words, if the critic predicts the label correctly, then he/she \emph{agrees} with the outcome label. On the other hand, if their prediction is incorrect, the critic \emph{disagrees} with the outcome. We also make an assumption that the outcome labels generated by COMPAS are presented to the critics.

Given the disagreement feedback from 400 critics, we evaluated COMPAS for different group fairness notions across both race and gender. Figure \ref{Fig: Estimated EO} demonstrates that the estimated equal opportunity (black line) from lower and upper bounds is close 0 as opposed to the true equal opportunity (blue line) for most of the critics. Moreover, a few violations in both upper and lower bounds can be observed. These violations may arise because of deviations from Bayes' rule in human behavior, which have been documented in the psychology literature \cite{Ouwersloot1998}. Similar to equal opportunity, the estimated predictive equality remains close to zero for majority of the critics. Figure \ref{Fig: Estimated PE} depicts that the estimated predictive equality follows the trends of the true predictive equality for majority of the critics. Unfortunately, number of violations increased in case of overall misclassification rates where, the predicted overall misclassification lies around 0.2.

Figure \ref{Fig: Error in estimation} demonstrates the mean absolute error of 400 critics across three different indefinite notions. In case of estimated equal opportunity, the error varies from about 5\% to 22\% with the mean error of 12\%. On the other hand, the mean error in estimating predictive equality is about 17\% with maximum error 30\%. Similarly, the mean absolute error in estimating overall misclassification is around 15\%. The minimum error rate (around 2.5\%) is similar for all three group fairness notions.

\section{Conclusions and Future Work}
In this paper, we proposed a novel feedback elicitation model for non-experts based disagreements. We identify two sets of groups fairness notions, one which can be precisely quantified from disagreements rates, and other which can be estimated based computed lower and upper bounds. Moreover, we validated our theoretical findings using real human feedback data across different group fairness notions and sensitive groups. In future, the objective to apply the proposed feedback elicitation model to kidney placement application by collecting actual disagreement feedback from patients and donors. Additionally, we hope to explore the relation between individual fairness and disagreements as well.

\bibliographystyle{splncs04}
\bibliography{sample-bibliography}

\begin{thebibliography}{10}
\providecommand{\url}[1]{\texttt{#1}}
\providecommand{\urlprefix}{URL }
\providecommand{\doi}[1]{https://doi.org/#1}

\bibitem{berk2018riskassess}
Berk, R., Heidari, H., Jabbari, S., Kearns, M., Roth, A.: Fairness in criminal
  justice risk assessments: The state of the art. Sociological Methods \&
  Research  \textbf{50}(1),  3--44 (2018)

\bibitem{bertsimas2017accept}
Bertsimas, D., Kung, J., Trichakis, N., Wojciechowski, D., Vagefi, P.A.: Accept
  or decline? an analytics-based decision tool for kidney offer evaluation.
  Transplantation  \textbf{101}(12),  2898--2904 (2017)

\bibitem{callender2008blacks}
Callender, C., Cherikh, W., Miles, P., Hermesch, A., Maddox, G., Nash, J.,
  Hernandez, A., Burston, B.: Blacks as donors for transplantation: suboptimal
  outcomes overcome by transplantation into other minorities. In:
  Transplantation proceedings. vol.~40, pp. 995--1000. Elsevier (2008)

\bibitem{caton2020fairness}
Caton, S., Haas, C.: Fairness in machine learning: A survey. arXiv preprint
  arXiv:2010.04053  (2020)

\bibitem{chouldechova2017fair}
Chouldechova, A.: Fair prediction with disparate impact: A study of bias in
  recidivism prediction instruments. Big data  \textbf{5}(2),  153--163 (2017)

\bibitem{chouldechova2018frontiers}
Chouldechova, A., Roth, A.: The frontiers of fairness in machine learning.
  arXiv preprint arXiv:1810.08810  (2018)

\bibitem{corbett2017algorithmic}
Corbett-Davies, S., Pierson, E., Feller, A., Goel, S., Huq, A.: Algorithmic
  decision making and the cost of fairness. In: Proceedings of the 23rd acm
  sigkdd international conference on knowledge discovery and data mining. pp.
  797--806 (2017)

\bibitem{denis2021fairness}
Denis, C., Elie, R., Hebiri, M., Hu, F.: Fairness guarantee in multi-class
  classification. arXiv preprint arXiv:2109.13642  (2021)

\bibitem{dressel2018accuracy}
Dressel, J., Farid, H.: The accuracy, fairness, and limits of predicting
  recidivism. Science advances  \textbf{4}(1),  eaao5580 (2018)

\bibitem{dwork2012fairness}
Dwork, C., Hardt, M., Pitassi, T., Reingold, O., Zemel, R.: {Fairness Through
  Awareness}. In: Proceedings of the 3rd innovations in theoretical computer
  science conference. pp. 214--226. ACM (2012)

\bibitem{gillen2018online}
Gillen, S., Jung, C., Kearns, M., Roth, A.: Online learning with an unknown
  fairness metric. In: Advances in Neural Information Processing Systems. pp.
  2600--2609 (2018)

\bibitem{grgic2018human}
Grgic-Hlaca, N., Redmiles, E.M., Gummadi, K.P., Weller, A.: Human perceptions
  of fairness in algorithmic decision making: A case study of criminal risk
  prediction. In: Proceedings of the 2018 World Wide Web Conference. pp.
  903--912 (2018)

\bibitem{MoritzOpportunities}
Hardt, M., Price, E., Srebro, N.: {Equality of Opportunity in Supervised
  Learning}. In: Lee, D.D., Sugiyama, M., Luxburg, U.V., Guyon, I., Garnett, R.
  (eds.) Advances in Neural Information Processing Systems 29, pp. 3315--3323.
  Curran Associates, Inc. (2016)

\bibitem{hardt2016equality}
Hardt, M., Price, E., Srebro, N., et~al.: {Equality of Opportunity in
  Supervised Learning}. In: {Advances in Neural Information Processing
  Systems}. pp. 3315--3323 (2016)

\bibitem{harrison2020empirical}
Harrison, G., Hanson, J., Jacinto, C., Ramirez, J., Ur, B.: An empirical study
  on the perceived fairness of realistic, imperfect machine learning models.
  In: Proceedings of the 2020 Conference on Fairness, Accountability, and
  Transparency. pp. 392--402 (2020)

\bibitem{hart2021optn}
Hart, A., Lentine, K., Smith, J., Miller, J., Skeans, M., Prentice, M.,
  Robinson, A., Foutz, J., Booker, S., Israni, A., et~al.: Optn/srtr 2019
  annual data report: kidney. American Journal of Transplantation  \textbf{21},
   21--137 (2021)

\bibitem{joseph2016fairness}
Joseph, M., Kearns, M., Morgenstern, J.H., Roth, A.: Fairness in learning:
  Classic and contextual bandits. Advances in neural information processing
  systems  \textbf{29} (2016)

\bibitem{jung2019eliciting}
Jung, C., Kearns, M., Neel, S., Roth, A., Stapleton, L., Wu, Z.S.: Eliciting
  and enforcing subjective individual fairness. arXiv preprint arXiv:1905.10660
   (2019)

\bibitem{liu2017calibrated}
Liu, Y., Radanovic, G., Dimitrakakis, C., Mandal, D., Parkes, D.C.: Calibrated
  fairness in bandits. arXiv preprint arXiv:1707.01875  (2017)

\bibitem{mcculloh2023experiment}
McCulloh, I., Stewart, D., Kiernan, K., Yazicioglu, F., Patsolic, H., Zinner,
  C., Mohan, S., Cartwright, L.: An experiment on the impact of predictive
  analytics on kidney offer acceptance decisions. American Journal of
  Transplantation  (2023)

\bibitem{mehrabi2021survey}
Mehrabi, N., Morstatter, F., Saxena, N., Lerman, K., Galstyan, A.: A survey on
  bias and fairness in machine learning. ACM Computing Surveys (CSUR)
  \textbf{54}(6),  1--35 (2021)

\bibitem{Ouwersloot1998}
Ouwersloot, H., Nijkamp, P., Rietveld, P.: {Errors in Probability Updating
  Behaviour: Measurement and Impact Analysis}. Journal of Economic Psychology
  \textbf{19}(5),  535--563 (1998)

\bibitem{pessach2020algorithmic}
Pessach, D., Shmueli, E.: Algorithmic fairness. arXiv preprint arXiv:2001.09784
   (2020)

\bibitem{rouzot2022learning}
Rouzot, J., Ferry, J., Huguet, M.J.: Learning optimal fair scoring systems for
  multi-class classification. In: ICTAI 2022-The 34th IEEE International
  Conference on Tools with Artificial Intelligence (2022)

\bibitem{saxena2019fairness}
Saxena, N.A., Huang, K., DeFilippis, E., Radanovic, G., Parkes, D.C., Liu, Y.:
  How do fairness definitions fare? examining public attitudes towards
  algorithmic definitions of fairness. In: Proceedings of the 2019 AAAI/ACM
  Conference on AI, Ethics, and Society. pp. 99--106 (2019)

\bibitem{sekercioglu2021machine}
Sekercioglu, N., Fu, R., Kim, S.J., Mitsakakis, N.: Machine learning for
  predicting long-term kidney allograft survival: a scoping review. Irish
  Journal of Medical Science (1971-)  \textbf{190},  807--817 (2021)

\bibitem{srivastava2019mathematical}
Srivastava, M., Heidari, H., Krause, A.: Mathematical notions vs. human
  perception of fairness: A descriptive approach to fairness for machine
  learning. In: Proceedings of the 25th ACM SIGKDD International Conference on
  Knowledge Discovery \& Data Mining. pp. 2459--2468 (2019)

\bibitem{threlkeld2021reducing}
Threlkeld, R., Ashiku, L., Canfield, C., Shank, D.B., Schnitzler, M.A.,
  Lentine, K.L., Axelrod, D.A., Battineni, A.C.R., Randall, H., Dagli, C.:
  Reducing kidney discard with artificial intelligence decision support: The
  need for a transdisciplinary systems approach. Current transplantation
  reports pp.~1--9 (2021)

\bibitem{yaghini2021human}
Yaghini, M., Krause, A., Heidari, H.: A human-in-the-loop framework to
  construct context-aware mathematical notions of outcome fairness. In:
  Proceedings of the 2021 AAAI/ACM Conference on AI, Ethics, and Society. pp.
  1023--1033 (2021)

\end{thebibliography}

\newpage

\appendix

\section{Ethical Implications of This Work}

This paper thrives towards improving diversity, equity, and inclusion (DEI) within fairness evaluation of machine learning algorithms. If successful, the paper enables us to collect feedback from diverse stakeholders in a given system, and use it in its maintenance and improvement. In addition, the simulation experiments conducted in this paper are designed using a publicly available dataset, which  is borrowed from the work carried out by Dressel and Farid in \cite{dressel2018accuracy}. The data released by Dressel and Farid is anonymized, and does not contain any human subject information regarding the critics. The data was itself curated from anonymous critics that were hired on Amazon's Mechanical Turk, which does not reveal any information regarding the crowd workers. Furthermore, since the original work by Dressel and Farid was approved by their Institutional Review Board at the Dartmouth College, as detailed in \cite{dressel2018accuracy}. 

\section{Reproducibility Statement}
All the code related to the simulation experiments stated in this paper will be available to the public along with the camera-ready version, once the paper is accepted for publication.

\end{document}